\documentclass[11pt]{article}

\usepackage{fullpage}
\usepackage{amsmath,amsfonts,amsthm}
\usepackage{amssymb}
\usepackage{mathrsfs}
\usepackage{cite}
\usepackage[ruled, linesnumbered, vlined, commentsnumbered]{algorithm2e}
\usepackage{array}
\usepackage[caption=false,font=normalsize,labelfont=sf,textfont=sf]{subfig}
\usepackage{textcomp}
\usepackage{stfloats}
\usepackage{url}
\usepackage{verbatim}
\usepackage{graphicx}
\usepackage{color}
\usepackage{balance}
\usepackage[flushleft]{threeparttable}

\def\BibTeX{{\rm B\kern-.05em{\sc i\kern-.025em b}\kern-.08em
    T\kern-.1667em\lower.7ex\hbox{E}\kern-.125emX}}
\usepackage{balance}
\usepackage{multirow}
\usepackage{hyperref}
\usepackage{bm}

\usepackage{geometry}
\newtheorem{remark}{\bf Remark}[section]
\newtheorem{definition}{\bf Definition}[section]
\newtheorem{assumption}{\bf Assumption}[section]
\newtheorem{lemma}{\bf Lemma}[section]
\newtheorem{proposition}{\bf Proposition}[section]
\newtheorem{theorem}{\bf Theorem}[section]

\newcommand{\Real}{\mathbb R}
\newcommand{\norm}[1]{\left\Vert#1\right\Vert}

\usepackage{fullpage}
\usepackage{authblk}
\usepackage{footnote}
\usepackage{prettyref}
\usepackage{setspace}
\usepackage{scrpage2}

\usepackage{lscape}
\usepackage{authblk}

\begin{document}
\title{\vspace{-1ex}\LARGE\textbf{Model-Assisted Probabilistic Safe Adaptive Control With Meta-Bayesian Learning}\footnote{This manuscript is submitted for potential publication. Reviewers can use this version in peer review.}}

\author[1]{\normalsize Shengbo Wang}
\author[2]{\normalsize Ke Li}
\author[3]{\normalsize Yin Yang}
\author[4]{\normalsize Yuting Cao}
\author[5]{\normalsize Tingwen Huang}
\author[6]{\normalsize Shiping Wen}
\affil[1]{\normalsize College of Computer Science and Engineering, University of Electronic Science and Technology of China, Chengdu 611731, China}
\affil[2]{\normalsize Department of Computer Science, University of Exeter, EX4 4QF, Exeter, UK}
\affil[34]{\normalsize  College of Science and Engineering, Hamad Bin Khalifa University, 5855, Doha, Qatar}
\affil[5]{\normalsize Science Program, Texas A \& M University at Qatar, Doha 23874, Qatar}
\affil[6]{\normalsize Australian AI Institute, Faculty of Engineering and Information Technology, University of Technology Sydney, NSW 2007, Australia }

\affil[$\ast$]{\normalsize Email: \texttt{shnbo.wang@foxmail.com},  
\texttt{k.li@exeter.ac.uk},
\texttt{shiping.wen@uts.edu.au}}

\date{}





\maketitle

\vspace{-3ex}
{\normalsize\textbf{Abstract: } Breaking safety constraints in control systems can lead to potential risks, resulting in unexpected costs or catastrophic damage. Nevertheless, uncertainty is ubiquitous, even among similar tasks. In this paper, we develop a novel adaptive safe control framework  that integrates meta learning, Bayesian models, and control barrier function (CBF) method. Specifically, with the help of CBF method, we learn the inherent and external uncertainties by a unified adaptive Bayesian linear regression (ABLR) models, which consists of a forward neural network (NN) and a Bayesian output layer. Meta learning techniques are leveraged to pre-train the NN weights and priors of the ABLR model using data collected from historical similar tasks. For a new control task, we refine the meta-learned models using a few samples, and introduce pessimistic confidence bounds into CBF constraints to ensure safe control. Moreover, we provide theoretical criteria to guarantee probabilistic safety during the control processes. To validate our approach, we conduct comparative experiments in various obstacle avoidance scenarios. The results demonstrate that our algorithm significantly improves the Bayesian model-based CBF method, and is capable for efficient safe exploration even with multiple uncertain constraints. }


\section{Introduction}
\label{sec:introduction}

Despite the existence of numerous designs, significant research efforts, and successful applications in the field of control systems, the development of a reliable and secure controller that combines robust theoretical foundations with exceptional performance continues to present a formidable challenge. This challenge has captured the attention of researchers from diverse fields, including robotics \cite{GarciaF15} and healthcare \cite{CoronatoNPP20}, among others. In the context of control systems, safety is evaluated based on the system state. In this study, we focus on \textit{probabilistic safe control}, wherein a safe controller is expected to prevent the system from entering hazardous states with an acceptable probability \cite{BerkenkampTS017, CastanedaCZTS21, WabersichHCZ22}. Due to the intricate nature of calculating the safe state space for a general dynamics-driven system, ensuring safety by designing or learning a safe controller is rather complex. Existing safe control strategies include model predictive control \cite{BrunkeGHYZPS22}, reachability analysis \cite{BansalCHT17}, and control barrier function (CBF) method \cite{AmesXGT17}. In our research, we build upon the CBF method, which ensures that the system state remains within safe regions by defining a forward invariant set. This set is a subset of the safe region and restricts the system state within its boundaries. Furthermore, we take into account the presence of uncertainty, which not only have a more significant impact on the system state than small disturbances \cite{Jankovic18}, and does not have an analytical format as well \cite{WangLWSZH23}. Prior work has addressed this issue in \cite{TaylorSYA20,FanNTAAT20,BrunkeZS22,DhimanKFA23} by introducing non-structured learning-based techniques. However, these methods, as we will discuss later, suffer from limitations in terms of real-time safe deployment efficiency and practicality in data-deficient online control tasks.

\par The learning-based safe control methods face a dilemma. To effectively model uncertainty, it is preferable to sample data randomly and independently from the control space. However, considering safety issues, driving an unknown system for data collection can lead to a high risk of unsafe behaviors. In addition, in the context of control system, data collection is typically based on a few state trajectories, which means the data is not independently distributed in most cases. This is exemplified in Fig. \ref{fig:illustrationexample}, where we aim to estimate the unsafe region by sampling data from a given safe trajectory. Previous works \cite{FanNTAAT20, DhimanKFA23} have utilized Gaussian processes (GPs) \cite{GPML} to recognize the safe set. However, this approach often leads to conservative estimations, as indicated by the large portion of the safe area labeled as 'unsafe' in the green covered region. To address this issue, researchers have explored the concept of safe exploration, which involves actively expanding system trajectories into unknown regions for better estimation \cite{KollerBT018,WachiSYO18,LiuSCAY20,ChakrabartyDCR22}. However, we argue that on-line exploration introduces additional time costs in the current control task. Moreover, it relies on a more ideal control environment than the practical one, such as the smoothness of the safety measurements, to minimize the risk of unsafe behaviors. It also struggles to ensure safety effectively when uncertainty exists.

We employ the concept of meta learning to effectively address inherent uncertainties in the current control task by leveraging empirical knowledge learnt from historical similar tasks \cite{FinnAL17, RothfussHCK21}. In the field of control systems, meta learning has been utilized for various purposes, such as modeling system dynamics \cite{HarrisonSP18}, conducting control design \cite{ArcariMSCFHZ23}, and adapting to new environments in the presence of uncertainties \cite{MajumdarFS21}. With regards to safety, meta learning can significantly enhance the estimation of unsafe regions, as depicted by the yellow covered region in Fig. \ref{fig:illustrationexample}. Despite its empirical effectiveness, there exists a theoretical gap that hinders the application of meta learning in facilitating safe control tasks \cite{ZhangCFLJ20,LewSHBP22,BennettMK23}. In this work, we aim to present a novel safe control framework combining the CBF method and meta learning techniques. Moreover, we will explore its theoretical criteria to ensure probabilistic safety mathematically.

\begin{figure}[tb]
    \centering
    \scalebox{.83}{\includegraphics{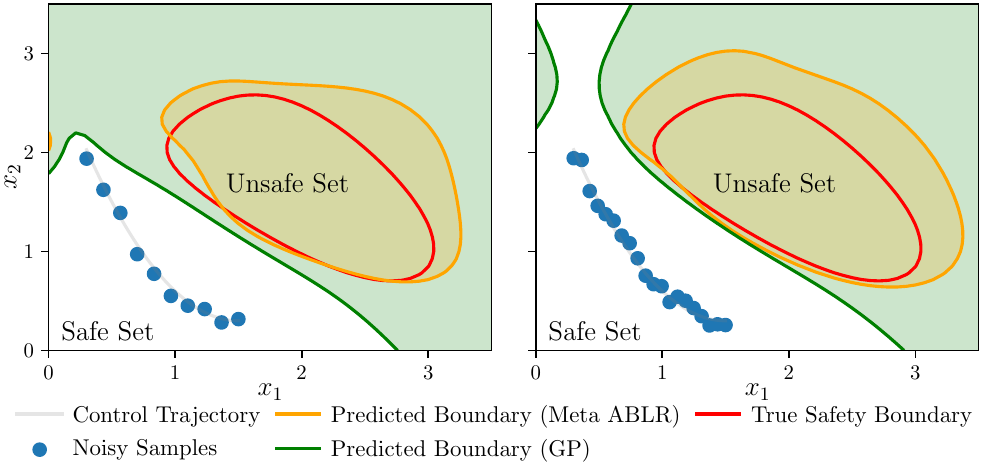}}
    \caption{Illustration of the meta-learning effect for safety boundary estimation. We construct a GP model with non-informative priors, and an ABLR model with meta-learned priors. Both models are trained on the same (Left) $10$ random samples, and (Right) $20$ random samples from a given trajectory. Please see Appendix \ref{appendix:illustration} for detailed settings and illustrations.
    }
    \label{fig:illustrationexample}
\end{figure}

\subsection{Related Works}
\label{section:relatedwork}


\subsubsection{CBF-based Adaptive Control Design}
\label{section:relatedwork_adaptivecbf}
In order to ensure safety, it is inevitable to introduce more conservativeness into a controller when model uncertainty exists. Taking the worst-case robustness into account results in the most conservative controllers \cite{Jankovic18}. This method is not appealing to researchers because of its infeasibility for general cases \cite{XiaoBC22}. To solve this problem in a more elegant way, it is necessary to further assumed that the uncertainty conforms to specific forms. 

For structured parametric uncertainty, i.e., the uncertainty only exists in some unknown parameters of a structured function, the ideas from the adaptive control design are adopted to obtain a better safe controller \cite{TaylorA20, XiongZTX23}. To achieve parameter identification simultaneously, an adaptive tuning law is developed for finite-time estimation \cite{BlackAP21}, which is further improved for a less conservative design in \cite{WangLWSZH23}. 
For \emph{unstructured uncertainty} as studied in this paper, accurate estimation or even point-wise safety may not be guaranteed. We roughly divide these works into two categories: (i) directly modelling the unknown dynamics, and (ii) merely reducing the uncertainty in CBF constraints. For the first case, (deep) neural networks (NNs) are developed to build effective learning frameworks empirically \cite{ChengOMB19}. For rigorous theoretical analysis, Bayesian methods are more attractive for safety-critical applications \cite{FanNTAAT20}. A principled Bayesian dynamics learning framework under CBF constraints is presented in \cite{DhimanKFA23}. However, these methods do not scale well on high dimensional systems and real-time control tasks due to high computational costs. For the latter case, NN-based control scheme is studied in \cite{TaylorSYA20} and \cite{WangMLS021} for one- and high-order CBF methods respectively. The critical limitation of these works lies in the lack of theoretical analysis. In addition, the framework in \cite{TaylorSYA20} needs episodic training and cannot work in on-line control tasks. From a Bayesian perspective, GPs are introduced to model the unknown parts of CBF constraints with data collected in an event-triggered manner to reduce the computational complexity \cite{CastanedaCZTS21, abs-2208-10733}. Similar to our method, the Bayesian linear regression (BLR) model is introduced to learn the unknown residual part of CBF constraints with point-wise guaranteed safety \cite{BrunkeZS22}. However, the basis functions used in \cite{BrunkeZS22} should be chosen from a set of class $\mathcal{K}_\infty$ functions, which is often hard to determine for accurate modeling and effective adaptation. Besides, all methods assume the unknown part is static, which is impractical in dynamic control environments.

\subsubsection{Meta Learning for Safe Control}
\label{section:relatedwork_meta}
For dynamic control scenarios where the control environment is changing, e.g. time-varying disturbances \cite{RichardsASP21} or switched external inputs \cite{MajumdarFS21}, the generalization ability of a controller is significant. In terms of CBF methods, the generalization ability can be enhanced by (i) the adaptive structure of CBFs considering dynamic environment \cite{XiaoBC22}, and (ii) the dynamic parameter in the CBF constraints \cite{XiaoBC20}. The former is achieved either automatically by machine learning approaches or manually by domain experts. The latter is related to parameter optimization, current methods including differential convex optimization techniques \cite{MaZTS22} and reinforcement learning \cite{abs-2303-04313}.
Nevertheless, by assuming perfect knowledge of environments, these methods can hardly work when uncertainty exists, which is an important issue to be addressed in this paper.

The meta-learning techniques are well-known capable of modeling uncertainties across similar tasks. By assuming shared implicit variables \cite{FinnAL17} or unified hyper-priors of learning models \cite{GrantFLDG18}, a meta model pre-trained with data of historical tasks can effectively adapt to a new task \cite{HarrisonSP18}. The meta-adaptive control strategy has been studied in \cite{RichardsASP21}. In \cite{MajumdarFS21}, the probably approximately correct (PAC) is integrated with Bayesian learning framework for provable generalization ability of the designed controllers. However, both studies does not work in safety-critical environments. Very recently, the safe meta-control algorithms are developed based on Bayesian learning and reachability analysis in \cite{LewSHBP22}, where unfortunately, estimation processes and safe control algorithms are rather complex. Differently, we are interested in estimating the scalar uncertainty in a CBF constraint, which is much simpler than estimating vector uncertainties for high-dimensional systems. Moreover, we argue that our method can adapt to changeable environments, while the work in \cite{LewSHBP22} only considers uncertain dynamics in a deterministic environment.

\subsection{Our Contributions}
In this work, we make three key contributions.
\begin{itemize}
    \item We propose a systematic development of an adaptive and probabilistic safe control framework by integrating meta learning techniques into the control barrier function (CBF) method. Our algorithm achieves computational efficiency in on-line control through the direct modeling of scalars in CBF constraints and the introduction of finite feature spaces based on adaptive BLR models.
    \item We provide a theoretical investigation into the criteria for ensuring probabilistic safety in our designed algorithm. Under mild assumptions, we establish upper bounds on the weight estimation errors and a relationship between these bounds and the predicted confidence levels, ensuring probabilistic safety during control. 
    \item We conduct empirical evaluations to assess the performance of our algorithm in various obstacle avoidance control scenarios. The results demonstrate the superior performance of our algorithm compared to robust and GP-based CBF methods, both with and without online sampling. Furthermore, our method exhibits enhanced efficiency in safe exploration through online sampling.
\end{itemize}

\subsection{Notations}
Throughout this paper, $\Real$ denotes the real value space, while $\Real^n$ denotes the $n-$dimensional real space. For a vector $x\in \Real^n$ and a differential scalar function $h(x)$, $\nabla_x h = \partial h(x)/ \partial x$. The Lie derivative of  $f(x)\in \Real^n$ w.r.t. $h(x)\in \Real$ is denoted by $\mathscr{L}_f h(x) = \nabla_x h^{\top} f(x)$. $\chi_d^2(p)$ represents the $p$-th quantile of the $\chi^2$ distribution with $d$ degree of freedom (DOF). For a matrix $A$, $\lambda_{\max}(A)$ and $\lambda_{\min}(A)$ denote the maximum and minimum eigenvalues of $A$, respectively. Let $\Vert x \Vert_A = \sqrt{x^\top A x}$ be the weighted $2$-norm of a vector $x$.

\section{Preliminaries}
\label{section:preliminaries}


\subsection{Problem Formulation}
\label{section:problem_formulation}

The autonomous control systems are considered to take the following \emph{control-affine} structure in this paper:
\begin{equation}
    \dot x = f(x) + g(x) u + \varphi_\omega(x) + \epsilon, \label{eqn:system_dynamics}
\end{equation}
where $x\in \Real^{n}$ represents the system state and $u\in \Real^{m}$ is the system input. Both drift dynamics $f: \Real^n\to\Real^n$ and input dynamics $g:\Real^n\to\Real^{n\times m}$ are known and locally Lipschitz continuous, corresponding to the prior knowledge of the systems. In addition, we introduce an unknown vector $\varphi_\omega:\Real^{n} \to \Real^{n}$. Here, $\omega$ denotes the label of a certain control environment. The unknown small noise $\epsilon$ is also considered due to ubiquitous disturbances or observation errors. This definition of our systems is also taken in other works \cite{LewSHBP22, MajumdarFS21, MaZTS22}. In terms of safety, the system state should stay in certain regions $\mathcal{X}_\omega \subset \Real^n$ at a high probability in case of potential risks. Note that $\mathcal{X}_\omega$ allows to change in different environments. The control input should be taken from a given set $\mathcal{U} \subset \Real^m$. We formalize the safe control problem as follows.

\emph{Probabilistic Safe Control Problem} (PSCP):
\begin{gather}
    \min_{u\in \mathcal{U}} ~ l(x, u) \label{eqn_ccscp} \\
    \text{subject to} \quad 
     \dot x = f(x) + g(x) u + \varphi(x, \omega) + \epsilon, \tag{\ref{eqn_ccscp}{a}} \label{eqn_ccscp_dynamics}\\
    \mathbb{P}\left( x \in \mathcal{X}_\omega \right) > 1 - \delta. \tag{\ref{eqn_ccscp}{b}} \label{eqn_ccscp_prob_safety}
\end{gather}
In the above equations, $\delta \in [0, 1)$ is a given threshold, and $l: \Real^{n\times m}\to\Real$ denotes a control objective, such as tracking a nominal input. PSCP is not uncommon in real-life applications. For example, a vehicle is expected to avoid moving on irregular areas of the road in terms of high mechanical wear and tear or other unpredictable risks, serving as a probabilistic safety constraint. While the resistance, load, road condition, are changing according to different scenarios, which are naturally uncertain a priori, leading to the unknown dynamics $\varphi$ and different admissible regions $\mathcal{X}_\omega$. We endeavor to propose a safe control method to solve PSCP. Before further investigations, we present a brief overview of CBF and adaptive BLR (ABLR) models with some necessary assumptions.

\subsection{Control Barrier Function Method}
\label{section:cbf_method}
The CBF is defined as a measure of the safety distance. Formally, a valid CBF is given as follows.
\begin{definition}[Control Barrier Function \cite{AmesXGT17}]
\label{definition:cbf}
Consider a continuously differentiable function $h_\omega(x):\Real^{n} \to \Real$, and a closed convex set defined by $\mathcal{C}_\omega = \left\{ x \vert h_\omega(x)\ge0 \right\}$. If $\mathcal{C}_\omega \subseteq \mathcal{X}_\omega$, then $h_\omega(x)$ is a valid CBF for systems \eqref{eqn:system_dynamics}.
\end{definition}

Dependent to $\omega$, a CBF may differ in changing environments. To ensure safety during the entire control procedures, the concept of invariant set is introduced to obtain a valid CBF-based safety constraint below.

\begin{definition}[Safety Constraint \cite{AmesXGT17}]
\label{definition:safety_cbf}
    For an initial state $x_0 \in \mathcal{C}_\omega$, i.e. $h_\omega(x_0) \ge 0$, if there exists a locally Lipschitz continuous controller satisfying
    \begin{equation}
    \sup_{u \in \mathcal{U}} \nabla_x h^\top_\omega(x) \dot x  \ge - \alpha (h_\omega(x)), \label{eqn:definition_CBF}
    \end{equation}
    where $\alpha$ represents an extended class $\mathcal{K}$ function, then $h_\omega(x)$ is forward invariant. As a consequence, systems \eqref{eqn:system_dynamics} are safe with probability $1$, i.e. $x \in  \mathcal{X}_\omega$ starting from $x_0$.
\end{definition}

For systems \eqref{eqn:system_dynamics}, CBF constraints \eqref{eqn:definition_CBF} are linear to $u$, making it tractable to compute a safe input. We assume the admissible control input set $\mathcal{U}$ is described by a linear matrix inequality as $u_{\min} \leq A u \leq u_{\max}$. When $A=I_m$, $\mathcal{U}$ represents the input saturation. Moreover, the loss function is designed quadratic to the input vector, such as $l = \norm{u - u_{\text{ref}}}^2$ with $u_{\text{ref}}$ the task specific nominal control input. In all, the safe controller is obtained by solving the following CBF-based quadratic programming (CBF-QP) problem at each control step:

\begin{gather}
    u^* = \min_{u\in \Real^m} ~ \norm{u - u_{\text{ref}}}^2 \label{eqn:cbfQP}\\
    \mathscr{L}_f h_\omega(x) + \mathscr{L}_g h_\omega(x) u  + \Delta_\omega(x) + \alpha \left(h_\omega(x)\right) \ge 0, \tag{\ref{eqn:cbfQP}{a}} \label{eqn:cbfQP_cbf}\\
    u_{\min} \leq A u \leq u_{\max}, \tag{\ref{eqn:cbfQP}{b}} \label{eqn:cbfQP_saturation}
\end{gather}
where $\Delta_\omega(x) = \mathscr{L}_{\varphi_{\omega}} h_\omega(x) + \epsilon_\omega$ with $\epsilon_\omega = \mathscr{L}_\epsilon h_\omega(x)$. Note that uncertainty exists in \eqref{eqn:cbfQP_cbf}. To tackle this, there are commonly two routines: (i) to find the largest bound of $\Delta_\omega(x)$ as the worst-case estimation \cite{Jankovic18}, and (ii) to estimate $\Delta_\omega(x)$ with adaptive error bounds \cite{TaylorA20,WangLWSZH23}. In this paper, the latter is chosen for better generalization. The following assumption is critical and widely made for estimation and adaptive design.

\begin{assumption} [Measurable Variables]
    \label{assumption:measurable}
    It is assumed that $\mathscr{L}_f h_\omega(x)$, $ \mathscr{L}_g h_\omega(x) u$, and $\dot h_\omega$ are all measurable.
\end{assumption}

The above statement is equal to assume $x$ is fully observable and the first derivative of $h_\omega$ can be computed or estimated. For simplicity, we denote the estimation of $\Delta_\omega(x)$ by $\tilde\Delta_\omega(x) = \dot h_\omega - \mathscr{L}_f h_\omega(x) -\mathscr{L}_g h_\omega(x) u$, and assume an \emph{unknown} parametric regressor as $ \tilde\Delta_\omega(x) = \phi^\top(x)\mathbf{w}^*$. The following assumption limits the regression error in theory. 
\begin{assumption} [Noise Bound \cite{YadkoriPS11}]
    \label{assumption:noisebound}
    For a sample sequence $\{\phi_t\}_{t=1}^{\infty}$ and a noise sequence $\{\eta_t\}_{t=1}^{\infty}$, the estimation noise $\eta_t = \Delta_\omega^t - \phi_t^\top \mathbf{w}^*$ is conditionally $\sigma_0$-sub-Gaussian, where $\sigma_0 > 0$ is a fixed constant. Thus for an integer $t>1$, there is
    \begin{equation}
        \forall \lambda \in \mathbb{R}, \quad \mathbb{E}\left[\exp^{\lambda \eta^t} \vert \phi_{1:t},  \eta_{1:t-1}\right] \leq \exp \left(\frac{\lambda^2 \sigma_0^2}{2}\right).
    \end{equation}
\end{assumption}

\begin{remark}
    Assumption \ref{assumption:noisebound} naturally assigns zero mean value and bounded variance to $\eta_t$. The zero-mean Gaussian noise $\mathcal{N}(0, \sigma_0^2)$ that is a common noise assumption in the research of observer-based control \cite{AugerHGMOK13} meets this condition.
\end{remark}

\subsection{Adaptive Bayesian Linear Regression}

To estimate the unknown part in \eqref{eqn:cbfQP_cbf}, we introduce the ABLR \cite{GPML}, a nonparametric model with tractable uncertainty quantification. A BLR models an unknown function $y(x)$ by $\tilde y(x) =   \phi^\top (x) \mathbf{w}$, where $\phi(x) = \left[\phi_1(x), \dots, \phi_D(x)\right]^\top$, and weights $\mathbf{w} \in \Real^D$. Instead of point estimation, we assume the weight priors follow a multivariate Gaussian distribution $\mathcal{N}\left( \bm{\mu}_0, \sigma_0^{2} K_0 \right)$. Given a set of data $\mathcal{D} = \left\{(x_i, y_i)\right\}_{i=1}^N$, applying the Bayes rule, the weight posteriors are computed by $p(\mathbf{w} \vert \mathcal{D}) = \mathcal{N}\left(\bm{\mu}_\tau, K_\tau\right)$ where
\begin{equation}
    \bm{\mu}_\tau = K_\tau \left(\Phi \mathbf{y} + K_0^{-1}\bm{\mu}_0 \right),  \quad K_\tau = \left(K_0^{-1}  + \Phi\Phi^\top \right)^{-1}.
    \label{eqn:posteriorBLR}
\end{equation}
In the above equation, $\Phi = \left[\phi(x_1), \dots, \phi(x_N) \right]$, and $\mathbf{y} = [y_1, \dots, y_N]^\top$. As a result, the posterior predictive distribution at a test point $x_t$ is given by
\begin{align}
    p(y \vert x_t, \mathcal{D}) = & \int p(y\vert x_t, \mathbf{w}) p(\mathbf{w}\vert \mathcal{D})\, d\mathbf{w} \nonumber\\
    = & \mathcal{N}\left(\bm{\mu}_\tau^\top \phi(x_t), \Sigma_t \right), \label{eqn:ABLR_predicted_distribution}
\end{align}
with
\begin{equation}
    \Sigma_t = \sigma_0^{2} \left(1  + \phi(x_t)^\top K_\tau \phi(x_t)\right).
    \label{eqn:predictedvarianceBLR}
\end{equation}

The BLR can adapt to different tasks if $\phi$ is able to represent the inductive biases among tasks with appropriate priors of $\mathbf{w}$. However, it is always difficult to determine a suitable combination of basis functions $\phi$ to strike a balance between prediction accuracy and computational efficiency \cite{Bishop07}. The ABLR \cite{SnoekRSKSSPPA15, Murphy12} shares the same structure as the vanilla BLR, differently, modifying the fixed $\phi(x)$ into a trainable mapping function, such as a (deep) NN \cite{PerroneJSA18} with $D$ outputs. Consequently, ABLR models require training to master insightful inductive biases as well as a good priors from data. The general structure of ABLR models is depicted in Fig. \ref{fig:overallalgorithm}.

\begin{remark}
    Compared to GPs \cite{GPML}, ABLR explicitly defines a kernel $k(x, x^\prime) = \phi^\top(x) K_0 \phi(x^\prime)$ in an finite dimensional feature space. This approximation benefits meta-Bayesian learning in two aspects: (i) the flexibility for adaptation by tuning kernel structures, (ii) the computational efficiency towards the number of collected data. Specifically, for $n$ data, the computation complexity is $\mathcal{O}(D^3n^2)$ for ABLR and $\mathcal{O}(n^3)$ for GPs \cite{Murphy12}. Since $D$ is fixed during control, ABLR is more efficient for large-scale on-line sampling.
\end{remark}

\section{Main Results}
\label{sec:mainresult}

\begin{figure*}[tb]
    \centering
    \scalebox{.35}
    {\includegraphics{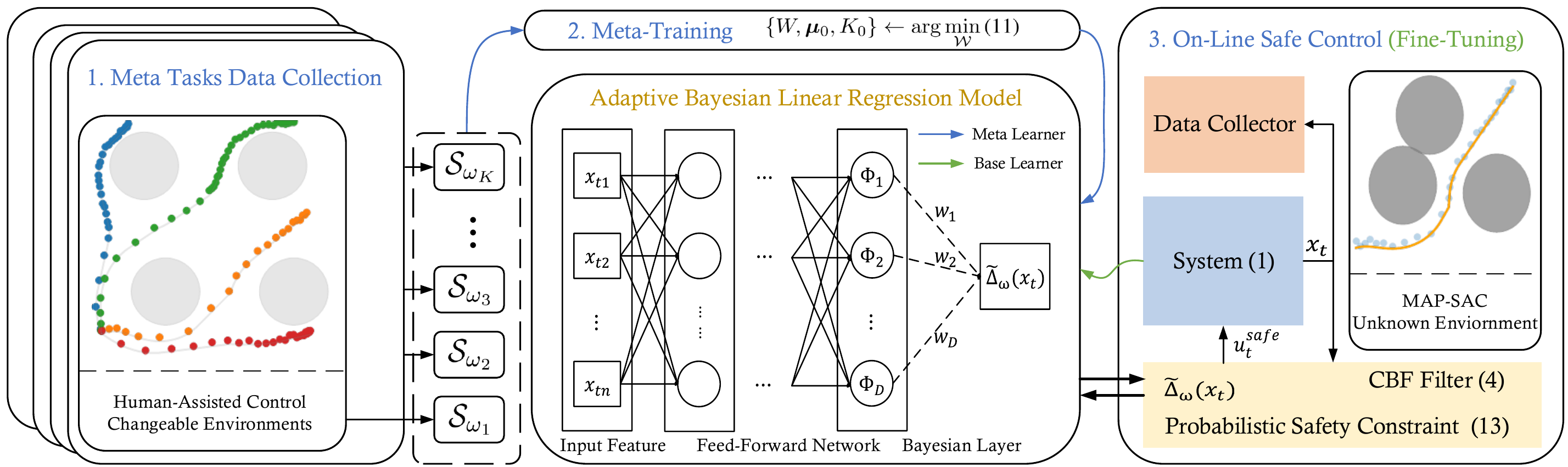}}
    \caption{Overview of the \texttt{MAP-SAC} framework.}
    \label{fig:overallalgorithm}
\end{figure*}

This section provides the systematic design and theoretical analysis of our model-assisted probabilistic safe adaptive control framework, dubbed \texttt{MAP-SAC}. We begin by proposing the overall framework in Section \ref{sec:algorithmdesign}. Then, we exploit its theoretical verification of safety in Section \ref{sec:theoreticalanalysis}. Finally in Section \ref{sec:implementation}, we present the detailed algorithms and practical implementation of \texttt{MAP-SAC}.

\subsection{Model-Assisted Probabilistic Safe Adaptive Control}
\label{sec:algorithmdesign}

The overall structure of the \texttt{MAP-SAC} is depicted in Fig. \ref{fig:overallalgorithm}. In what follows, we investigate the off-line meta-learning design of the ABLR model and the on-line probabilistic safe control framework based on CBF method.

\subsubsection{Meta-training an ABLR model}

The parameters of an ABLR model are learned under the framework of meta-learning, which contains a 'meta-learner' that extracts knowledge from many sampled trajectories of historical control tasks, and a 'base-learner' that can efficiently adapt to new tasks in an on-line manner \cite{FinnAL17}. For a Bayesian model, it is commonly to first learn the priors from historical data and then exploit them for the new tasks \cite{GrantFLDG18,MajumdarFS21}. 

The mapping functions, i.e. the feed-forward network $\phi$, and the priors of the Bayesian layer, i.e. $\mathcal{N}\left( \bm{\mu}_0, \sigma_0^{2} K_0 \right)$, are trained together. For brevity, we use a fixed network structure and put focus on learning the values of its weights, denoted by $W$. In all, the trainable parameters of an ABLR model is given by $\mathcal{W} = \{W, \bm{\mu}_0, K_0\}$. Denote the set of $K$ historical tasks by $\mathcal{T} = \left\{T_{\omega_i}(\xi_i)\right\}_{i=1}^{K}$ where $\xi_i \sim p(\xi)$. In each task, the control trajectories are sampled from environment labeled by $\omega_i$. It is assumed that uncertainties across tasks follow a fixed unknown distribution $p(\xi)$. For the $i$-th task, let the collected datasets be $\mathcal{S}_{\omega_i} = \{ (x^j, \tilde\Delta^j_{\omega_i}) \}_{j=1}^{t_i}$ where $t_i$ denotes the total number of data pairs. Note that $\tilde\Delta_{\omega_i}$ is an estimation of the true uncertainty $\Delta_{\omega_i}$. To train our model, the Kullback-Leibler divergence between the true model $p^*({\Delta}_t \vert x_t, \mathcal{S}_{\omega_i})$ and ABLR $p(\tilde{\Delta}_t \vert x_t, \mathcal{S}_{\omega_i})$ is minimized. This is equivalent to minimize the negative log likelihood \cite{Murphy12} given by
\begin{equation}
    L(\mathcal{W}) = - \mathbb{E}_{x, \tilde{\Delta}, \mathcal{S} \sim p(x, \Delta, \mathcal{S} \vert \xi)} \log p(\tilde{\Delta}_t \vert x_t, \mathcal{S}_{\omega}).
\end{equation}
Moreover, thanks to the Gaussian priors, the log likelihood can be computed analytically. Using Monte-Carlo estimation and substituting \eqref{eqn:ABLR_predicted_distribution}, we have
\begin{align}
    L(\mathcal{W}) \propto & \sum_{i=1}^{K}\sum_{j=1}^{t_i} \log \det \Sigma_i^j  + {\bm{\mu}_i^j}^\top {\Sigma_i^{j}}^{-1}  \bm{\mu}_i^j, \label{eqn:mc_loss}
\end{align}
in which $ \bm{\mu}_i^j = \tilde{\Delta}^j_{\omega_i} - \bm{\mu}_i^\top \phi(x^j)$, $\bm{\mu}_j$ and $\Sigma_i^j$ are the predicted mean and variance at point $x^j$ according to \eqref{eqn:posteriorBLR} and \eqref{eqn:predictedvarianceBLR}. We train $\mathcal{W}$ by minimizing \eqref{eqn:mc_loss} using historical tasks $T_{\omega_i}$ and trajectory data $\mathcal{S}_{\omega_i}$, $i=1,\dots,K$, in the 'meta-learner' phase. For a new task ${T}_{\omega_{K+1}}$, we collect trajectory data $\mathcal{S}_{\omega_{K+1}}$ online. As designed in \cite{HarrisonSP18}, only weight priors in Bayesian layer are adjusted for fine-tuning in the 'base-learner' phase.

\subsubsection{Fine-Tuning the meta-learned ABLR model}
The ABLR model obtained from the meta-learning stage is able to adapt to new tasks by fine-tuning with sequentially collected on-line data $\mathcal{S}^t_{\omega_{K+1}} = \{ (x^j, \tilde\Delta^j_{\omega_{K+1}}) \}_{j=1}^{t}$ in which $t$ is the number of data. In this stage, we only adjust the prior mean and covariance of $\mathbf{w}$, expecting that the forward network $\phi(x)$ has captured informative features among tasks from diverse historical data. According to \eqref{eqn:posteriorBLR}, the posteriors are given by:
\begin{align}
    \begin{split}
        \bm{\mu}_t = & K_{t-1} \left(\Phi_{t-1} \mathbf{y}_{t-1} + K_0^{-1}\bm{\mu}_0 \right),  \\
    K_{t} = & \left(K_0^{-1} + \Phi_{t-1}\Phi_{t-1}^\top \right)^{-1}, \\
    \Phi_t = & [\phi(x^1), \dots, \phi(x^{t-1}) ], \\
    \mathbf{y}_t = & [\tilde\Delta^1_{\omega_{K+1}}, \dots, \tilde\Delta^{t-1}_{\omega_{K+1}}]^\top.
    \end{split}
    \label{eqn:updateonline}
\end{align}

\subsubsection{Integrating the ABLR model in safe control framework}
For a Bayesian model, the confidence interval drawn from its predictive distribution can provide insightful information for the learning quality. Under informative priors through meta-learning, the ABLR model can make reasonable predictions of the uncertainty term $\Delta_{\omega}$ for on-line safe control. For simplification, denote the predicted Gaussian distribution of the ABLR model at a test point $x_t$ by $\mathcal{N}(\tilde\mu_t, \tilde\sigma_t^2)$ where $\tilde\mu_t$ and $\tilde\sigma_t$ is the mean and standard deviation of \eqref{eqn:ABLR_predicted_distribution}. The confidence interval related to some confidence level $\beta>0$ is computed as $\mathcal{I}_t = \left[ \tilde\mu_t - \beta\tilde\sigma_t, \mu_t + \beta\tilde\sigma_t \right]$. Therefore, we can integrate this into a modified CBF constraint below to formulate a probabilistic form of a CBF constraint.
\begin{definition}[Probabilistic Safety Constraint \cite{CastanedaCZTS21}]
    For a valid CBF defined in Definition \ref{definition:safety_cbf} and for a Bayesian model such as ABLR, the probabilistic safety constraint is given as
    \begin{equation}
        \mathscr{L}_f h_\omega(x) + \mathscr{L}_g h_\omega(x) u  + \tilde\mu(x) - \beta\tilde\sigma(x) + \alpha \left(h_\omega(x)\right) \ge 0. \label{eqn_pcbfQP_cbf}
    \end{equation}
\end{definition}

In this vein, we can design the CBF-based probabilistic safe controller by substituting \eqref{eqn_pcbfQP_cbf} with \eqref{eqn:cbfQP_cbf} in \eqref{eqn:cbfQP}. Before deployment, we provide a theoretical investigation to show on what degree \texttt{MAP-SAC} can ensure probabilistic safety.

\subsection{Theoretical Analysis and Probabilistic Safety}
\label{sec:theoreticalanalysis}

We begin the analysis by evaluate the modeling effect of ABLR models through meta-learning. The following assumptions are necessary as used in literature \cite{LewSHBP22}. Note that these assumptions correspond to the smoothness and correct model selection assumptions of GPs in function spaces \cite{GPML, SrinivasKKS10}.
\begin{assumption}[Capacity of Meta Models] \label{assumpstion:capacity}
    With Assumption \ref{assumption:noisebound}, $\forall \tilde \xi \sim p(\xi)$, there exists $\mathbf{w}^*(\tilde\xi) \in \mathbb{R}^D$ satisfying
    \begin{equation}
        \Delta_\omega(x) = \phi^\top(x)\mathbf{w}^*(\tilde\xi) + \eta(\tilde\xi), \quad \forall x\in \mathcal{X}_\omega .
        \label{eqn:optimalweight}
    \end{equation}
\end{assumption}
\begin{assumption}[Calibration of Meta Priors] \label{assumpstion:calibration}
    For $\xi \sim p(\xi)$, the error between $\mathbf{w}^*$ and prior mean $\bm{\mu}_0$ is calibrated with probability at least $1-\tilde\delta$, where $\tilde \delta = \delta/\kappa$ and $\kappa > 0$, as
    \begin{equation}
        \mathbb{P}\left( \Vert \mathbf{w}^* - \bm{\mu}_0 \Vert^2_{K_0^{-1}} < \sigma_0^2 \chi_D^2(1 - \tilde\delta) \right) \ge 1 - \tilde\delta.
    \end{equation}
\end{assumption}
\begin{remark}
    Theoretically, Assumption \ref{assumpstion:capacity} holds with the uniform approximation ability of (multi-layer) NNs \cite{HornikSW89}. It can be further relaxed to permit a bounded approximate error as stated in \cite{YadkoriPS11,LewSHBP22}. To ensure Assumption \ref{assumpstion:calibration}, we can set a large enough $\sigma_0$ for non-informative priors, leading to a conservative theoretical analysis of the meta model.
\end{remark}

We can now quantify the modeling effect of ABLR models by the following statements.

\begin{proposition}[Confidence Ellipsoid \cite{LewSHBP22,YadkoriPS11}] 
    \label{proposition:confidenceellipsoid}
    Let Assumptions \ref{assumption:measurable}, \ref{assumption:noisebound}, \ref{assumpstion:capacity} and \ref{assumpstion:calibration} hold. The meta ABLR model is trained by \eqref{eqn:mc_loss} with historical data and updates according to \eqref{eqn:updateonline} with on-line collected data $\mathcal{S}^t_{\omega_{K+1}}$. Let also $\mathbf{w}^*$ be the optimal weight as \eqref{eqn:optimalweight}. Then, $\mathbf{w}^*$ lies in the following set with probability at least $1-\tilde{\delta}$, as
    \begin{equation}
        \mathcal{C}^\delta_t = \left\{ \mathbf{w} \in \mathbb{R}^D \vert \Vert \mathbf{w} - \bm{\mu}_t \Vert_{K_t^{-1}} \leq \sigma_0 \Gamma^{\tilde \delta} \right\},
        \label{eqn:weightupperbound}
    \end{equation}
    where 
    \begin{align}
        \begin{split}
        \Gamma_t^{\tilde \delta} = &\sqrt{2\log \left(\frac{\det (K_t)^{-\frac{1}{2}} \det(K_0)^{\frac{1}{2}}}{\tilde\delta }\right)} \\
        & + \sqrt{\frac{\lambda_{\max}(K_t)}{\lambda_{\min}(K_0)} \chi^2_D (1 - \tilde\delta)}.
        \end{split}
    \end{align}
\end{proposition}
\begin{proof}
    The derivation is presented in Appendix \ref{appendix:proof}.
\end{proof}

The above statements quantifies the estimation error bounds of parameters in the ABLR model based on the meta-learned priors ($K_0$), on-line collected data ($K_t$), and user-specific probability ($\tilde\delta$). In the context of CBF method, we can combine the worst-case adaptive CBF method to obtain a valid CBF constraint that ensure safety with a specified probability. Based upon \eqref{eqn_pcbfQP_cbf}, it is better to determine the value of $\beta$ to correctly obtain the safety probability, other than believing the uncalibrated confidence intevels predicted by ABLR models.

\begin{theorem}[Probabilistic Safety]
\label{theorem:probabilitysafety}
    Let Assumptions \ref{assumption:measurable}, \ref{assumption:noisebound}, \ref{assumpstion:capacity} and \ref{assumpstion:calibration} hold. The meta ABLR model is trained by \eqref{eqn:mc_loss} with historical data and updates according to \eqref{eqn:updateonline} with on-line collected data $\mathcal{S}^t_{\omega_{K+1}}$. Then, the input $u$ satisfying probabilistic safety constraint \eqref{eqn_pcbfQP_cbf} renders systems \eqref{eqn:system_dynamics} safe with at least probability $1-\tilde\delta$ if we set $\beta \ge \Gamma_t^{\tilde \delta}$.
\end{theorem}
\begin{proof}
    The derivation is presented in Appendix \ref{appendix:proof}.
\end{proof}

\subsection{Algorithm Implementation}
\label{sec:implementation}

\begin{algorithm}[t!]
\SetKwFunction{Gradient}{Gradient}

\KwData{$\mathcal{S}_{\omega_i}$, $i\leftarrow 1$ \KwTo $K$, sampled from $K$ similar tasks}
\KwIn{Configurations of ABLR, training epoch $N$}
\KwOut{A pre-trained ABLR model for similar tasks}
Initialize an ABLR model with $\mathcal{W} = \{W, \bm{\mu}_0, K_0\}$\;
\For{$n\leftarrow 1$ \KwTo $N$}
{   
     \tcc{Calculate the meta loss \eqref{eqn:mc_loss}}
     $L(\mathcal{W}) \gets 0$\;
    \For{$i = 1, \dots K$}
    {
        $L(\mathcal{W}) += \sum_{j=1}^{t_i} \log \det \Sigma_i^j  + {\bm{\mu}_i^j}^\top {\Sigma_i^{j}}^{-1}  \bm{\mu}_i^j$\;
    }
    Update $\mathcal{W}$ by \Gradient{$L(\mathcal{W})$}\;
}
Save $\mathcal{W}$ and the ABLR configurations.
\caption{Meta learning stage of \texttt{MAP-SAC}}
\label{alg:map-sac-metatraining}
\end{algorithm}

\begin{algorithm}[t!]
\SetKwFunction{Gradient}{Gradient}

\KwData{A small $\mathcal{S}_{\omega_{K+1}}$ from current task}
\KwIn{Nominal dynamics and controllers; constraint number $N_c$; fine-tuning steps $N_f$, control steps $N_s$, sampling $F_s$, CBF-QP configurations}
\KwOut{Control trajectories with probabilistic safety}
\tcc{Warm start the learning models}
\For{$i\leftarrow 1$ \KwTo $N_c$}
{   
    Construct an ABLR model for the $i$-th constraint\;
    Load $\mathcal{W}$ and the $i$-th piece of $\mathcal{S}_{\omega_{K+1}}$ by \eqref{eqn:posteriorBLR}\;
    
    \For(\tcp*[f]{Fine-tuning}){$n\leftarrow 1$ \KwTo $N_f$}{
    $L(\bm{\mu}_0) \gets \sum_{j=1}^{t_{K+1}} {\bm{\mu}_{K+1}^j}^\top {\Sigma_{K+1}^{j}}^{-1}  \bm{\mu}_{K+1}^j$\;
    Update $\bm{\mu}_0$ by \Gradient{ $L(\bm{\mu}_0)$}\;
    }
}

\tcc{On-line safe control}
\For{$i\leftarrow 1$ \KwTo $N_s$}
{   
    \If(\tcp*[f]{On-line Sampling}){$i \%  F_s \leq 0$}{
        Sample state and observations, update $\mathcal{S}_{\omega_{K+1}}$\;
        Update posteriors according to \eqref{eqn:updateonline}\;
        Fine tune $\bm{\mu}_0$ by \Gradient{$L(\bm{\mu}_0)$}\;
    }
    Predict posterior means and variances by \eqref{eqn:ABLR_predicted_distribution}\;
    Calculate the safe control input by \eqref{eqn_pcbfQP_cbf}\;
    Input the safe controller to systems \eqref{eqn:system_dynamics}\;
}

\caption{On-line control stage of \texttt{MAP-SAC}}
\label{alg:mapsac-onlinecontrol}
\end{algorithm}

In this section, we discuss some implementation details of \texttt{MAP-SAC}. The pseudo code for two stages of \texttt{MAP-SAC} is presented in Algorithm \ref{alg:map-sac-metatraining} and Algorithm \ref{alg:mapsac-onlinecontrol}. 
Our discussions consist of the following aspects.

\subsubsection{Meta-task data collection} Although similar tasks are accessible in many control scenarios \cite{HarrisonSP18,BrunkeGHYZPS22,MajumdarFS21, RichardsASP21}, from which sampling is intractable considering safety issues. Therefore, the human-assisted control (HAC) is needed, to either inspect control procedures \cite{SaundersSSE18}, or pre-plan safe trajectories with tracking controllers \cite{ChengOMB19}. In addition, we can use high-fidelity simulations to avoid catastrophic unsafe behaviors \cite{PengAZA18}. For all examples considered in this work, we pre-plan safe trajectories for all meta control tasks.

\subsubsection{ABLR configurations and Meta training} The structure of forward NNs and cell number $D$ of the Bayesian layer determines the capacity of an ABLR model. To meet Assumption \ref{assumpstion:capacity}, we use multi-layer perceptron (MLP) in our experiments. We fix the prior covariance matrix $\sigma_0^2K_0$ during training in order to meet Assumption \ref{assumpstion:calibration} and reduce computation complexity, referring to \cite{HarrisonSP18} otherwise. In a more practical perspective, we can added regularization terms during training in case of over-fitting for better performance \cite{LewSHBP22}.

\subsubsection{Warm start the learning models} Since meta data are also expensive to collected, it is likely to have biased ABLR model after meta-training. We refine the ABLR model using data from the current task. In this stage, we only update the prior mean vector $\bm{\mu}_0$ \cite{SnoekRSKSSPPA15} by minimizing the simplified loss 
\begin{align}
    L(\bm{\mu}_0) \propto \sum_{j=1}^{t_{K+1}} {\bm{\mu}_{K+1}^j}^\top {\Sigma_{K+1}^{j}}^{-1}  \bm{\mu}_{K+1}^j,
    \label{eqn:simplifiedloss_refine}
\end{align}
which is similar to optimize the hyperparameters of a GP model while keeping kernel structure fixed. We can further speed up the optimization processes by manually computing the analytical format of gradients \cite{GPML}. Note that at least one sample of the current task is needed to warm start.

\subsubsection{Online sampling and control} Despite theoretically attractive, the lower bound of $\beta$ in Theorem \ref{theorem:probabilitysafety} is too conservative in practice. Typically, we can set $\beta=1.96$ for a $95\%$ confidence level in \eqref{eqn_pcbfQP_cbf}. In addition, we can refine the ABLR models using \eqref{eqn:simplifiedloss_refine} and safe-exploration techniques \cite{KollerBT018,WachiSYO18,LiuSCAY20} for better control performance. Furthermore, a better deployment for more efficient real-time computation can be a promising improvement of our algorithm \cite{WangCGYWH21}.

\section{Algorithm Validation}
\label{sec:experiemnts}
\subsection{Experiment Setup}
We test our framework on the obstacle avoidance control of a moving point, which is a popular verification application of safe control methods \cite{WangLWSZH23, BlackAP21, LewSHBP22}. We consider multiple scenarios, (i) uncertain dynamics with one fixed obstacle, (ii) uncertain dynamics with one uncertain obstacle, and (iii) uncertain dynamics with multiple uncertain obstacles. Our algorithm is compared with three methods: the optimal CBF method (\texttt{CBF-OPT}) with perfect knowledge of uncertainty \cite{AmesXGT17} (no uncertainty), the robust CBF method (\texttt{CBF-RUST}) with the worst-case estimation of uncertainty \cite{Jankovic18} (no learning), and the GP-based probabilistic CBF method (\texttt{GP2-SAC}) that model uncertainty by sampling data from the current task \cite{CastanedaCZTS21} (no meta). We expect \texttt{MAP-SAC} performs nearly the same as \texttt{CBF-OPT}, and significantly outperforms \texttt{GP2-SAC} with non-informative priors and \texttt{CBF-RUST}.

We implement GPs using GPy \cite{gpy2014} with Mat\'ern ${5/2}$ kernel, constant mean function, and fixed-noise Gaussian likelihood. We implement the ABLR model by JAX \cite{jax2018github}, with a $3$-layer multi-output NN in Haiku \cite{jax2018github}. Specifically, he first two layers of the NN are fully connected by $256$ $\tanh$ cells, while the last layer contains $D$ output sigmoid cells. For both models, we fixed the variance of noise distribution by $0.1$. The CBF-QP is based on the solver for convex optimization in \cite{DiamondB16}.

The dynamics of a moving point system is given by:
\begin{equation}
    \dot x = \left[ \begin{array}{cc}
         1 & 0 \\
         0 & 1 
    \end{array} \right] u - \underbrace {w * \left[ \begin{array}{c}
         \cos(x_1) + \sin(x_2) \\
         \cos(x_1) + \sin(x_2)
    \end{array} \right]}_{\varphi_\omega(x)},
    \label{eqn:movingpoint}
\end{equation}
where $\varphi_\omega(x)$ is unknown and $\omega$ takes values randomly from a known distribution $p(\xi)$. We consider a uniform distribution of $\omega$ as $\omega \sim {U}(0.5, 2.5)$. For simplicity, the obstacles are described by circles. The $i$-th obstacle is centered at $x_\omega^i, y_\omega^i$ with radius $r_\omega^i$. Therefore, a valid CBF for $i$-th obstacle is 
\begin{equation}
    h^i_\omega (x) = (x_1 - x_\omega^i)^2 + (x_2 - y_\omega^i)^2 - {r^i_\omega}^2.
\end{equation}
In the current task, the moving point starts from the origin and moves towards a target point $x_T = [3.0, 4.0]^\top$ with a nominal controller $u_{\text{ref}} = -k_f (x - x_T)$ with a fixed parameter $\omega = 1.5$. There are already $20$ noisy samples from a known safe trajectory, as the gray line in Fig. \ref{fig:illustrationexample}. The GP and ABLR models are optimized/fine-tuned according to these samples. For all scenarios, we use the following configurations: $k_p=10.0$, $\alpha=1.0$, $u_{\max}= 5$, $u_{\min} = -5.0$, and $\beta=1.96$. 

\subsection{Experiments Results}
\subsubsection{Uncertainty in Dynamics} We consider a fixed obstacle described by $x_r=y_r = 1.5$, and $r = 0.8$. We randomly sample $20$ different $\omega$ as historical tasks, and generate $30$ sparse samples in each task. We set $D=10$. The ABLR model is meta-trained according to these samples. The algorithms are evaluated on two simulations without/with on-line samplings respectively. To evaluate performance, we introduce the criteria from the field of adaptive CBF methods \cite{TaylorA20,BlackAP21,WangLWSZH23}. Therein, a better algorithm commonly admits a closer distance between system state and the obstacle, namely a smaller $h(x)$ during control. Figure \ref{fig:withoutonlinesampling} shows the different trajectories, when the on-line sampling is forbidden, i.e., both GP and ABLR models do not update during control. It is shown that \texttt{MAP-SAC} performs the best, while other two algorithms are more conservative. When online sampling is permitted, we set the sampling frequency by $10\mathrm{Hz}$. As shown in Fig. \ref{fig:withonlinesampling}, the performance of \texttt{GP2-SAC} improves a lot, while \texttt{MAP-SAC} is almost the same as the optimal controller \texttt{CBF-OPT}.

\begin{figure}[tb]
    \centering
    \scalebox{.8}
    {\includegraphics{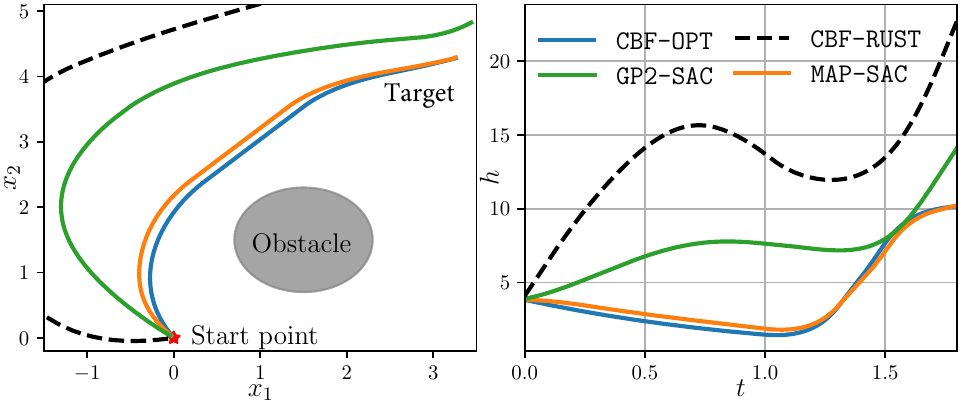}}
    \caption{Comparative results of the moving point experiment without on-line sampling. Both models do not update during control.}
    \label{fig:withoutonlinesampling}
\end{figure}

\begin{figure}[tb]
    \centering
    \scalebox{.8}{\includegraphics{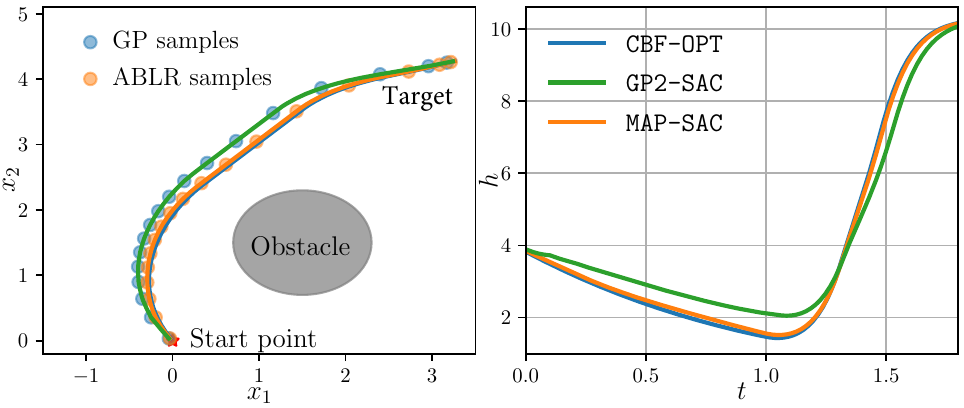}}
    \caption{Comparative results of the moving point experiment with on-line sampling. The GP model is re-optimized as long as a new sample is obtained.}
    \label{fig:withonlinesampling}
\end{figure}

\subsubsection{Uncertainty in Both Dynamics and an Obstacle} In this scenario, the obstacle in each meta task is generated randomly from $x_r, y_r\sim U(1.0, 4.0)$, $r\sim U(0.2, 1.0)$. Since uncertainty increases, we set $D=20$ and generate $50$ meta-tasks, each with $30$ sparse samples, to train the ABLR model. In the current task, the new obstacle locates on $x_r=1.0$, $y_r = 2.5$, and $r = 1.2$ (out-of-distribution). The simulation results are presented in Fig. \ref{fig:withoutonlinesampling_e2}, which are consistent with the first scenario.

\begin{figure}[tb]
    \centering
    \scalebox{.8}
    {\includegraphics{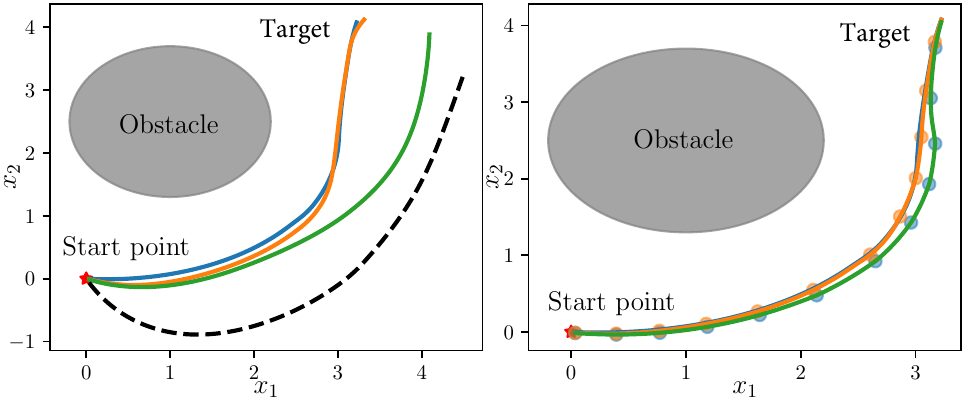}}
    \caption{Comparative results of the moving point experiment without/with on-line sampling. The uncertainties exist in both dynamics and obstacles.}
    \label{fig:withoutonlinesampling_e2}
\end{figure}

\subsubsection{Uncertainty in Multiple Obstacles} Note that the saved model configurations from the second scenario can directly used in the final experiment. We consider three obstacles located on $x_r=[1.0, 3.0, 2.5]$, $y_r = [2.5, 2.0, 0.5]$, and $r = [1.0, 0.5, 0.8]$. We build three ABLR models from the saved configurations, and refine them according to different pieces of data (observations are organized in pieces). We introduce three CBF constraints into CBF-QP for safe control. Without on-line sampling, as shown in Fig. \ref{fig:multiple_comparison}, both \texttt{GP2-SAC} and \texttt{CBF-RUST} fail to find a path to the target position due to conservative estimation of uncertainties, while \texttt{MAP-SAC} performs significantly better. With on-line samples, seen from Fig. \ref{fig:multiple_online}, \texttt{GP2-SAC} starts to explore a valid path, however, in a very inefficient manner. While \texttt{MAP-SAC} performs approximately the same as \texttt{CBF-OPT} but inefficient to some extend.

\begin{figure}[tb]
    \centering
    \scalebox{.8}{\includegraphics{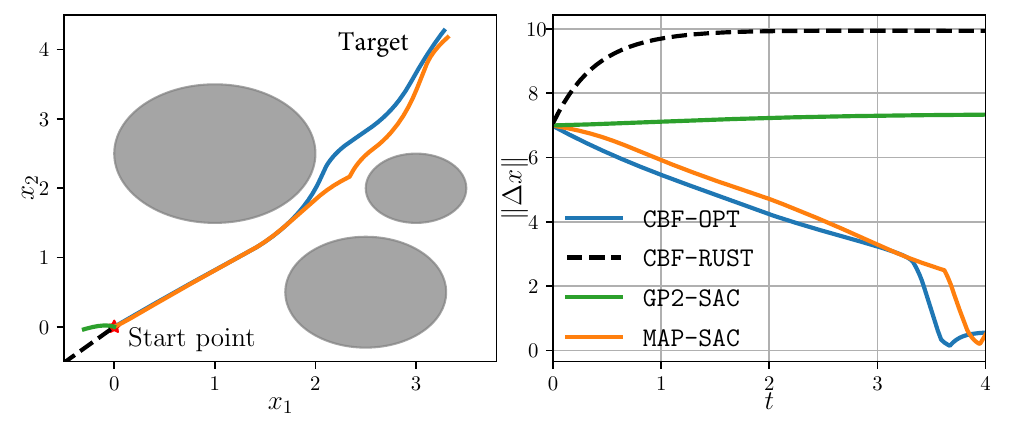}}
    \caption{Comparative results of the moving point experiment without on-line sampling. The uncertainties exist in dynamics and multiple obstacles.}
    \label{fig:multiple_comparison}
\end{figure}

\begin{figure}[tb]
    \centering
    \scalebox{.8}{\includegraphics{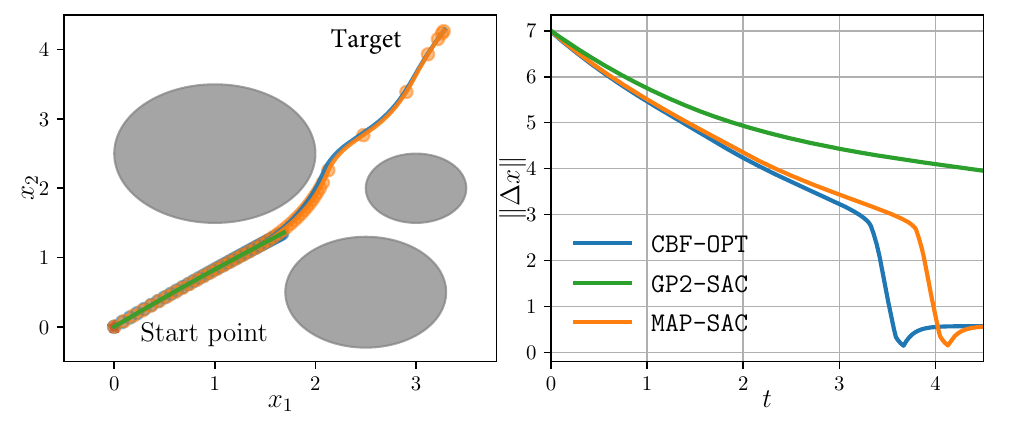}}
    \caption{Comparative results of the moving point experiment with on-line sampling. The uncertainties exist in dynamics and multiple obstacles.}
    \label{fig:multiple_online}
\end{figure}

\begin{remark}
    In Fig. \ref{fig:multiple_comparison} and \ref{fig:multiple_online}, $\Delta x = x - x_T$. We evaluate the efficiency of different algorithms by the error bound of $\Delta x$. In two simulations, \texttt{MAP-SAC} without on-line sampling is more efficient than it with new samples. It is mainly because of the rough, biased estimation of the uncertainty, with more risks of performing unsafe behaviors. The different efficiency between \texttt{MAP-SAC} and \texttt{GP2-SAC} is also revealed in Fig. \ref{fig:illustrationexample} As the GP model predicts a more conservative reachable set than meta-trained ABLR model at each step, \texttt{MAP-SAC} can be a more suitable choice for on-line safe exploration \cite{KollerBT018,WachiSYO18,LiuSCAY20}.
\end{remark}


\section{Conclusions}

In this work, we develop a novel framework for safe control against uncertainties by leveraging Bayesian learning models with meta learning techniques. The theoretical analysis establishes the probabilistic safety guarantee of our method, while also providing insights into the implementation details of the algorithm. We conduct comparative experiments on several control scenarios, demonstrating the effectiveness and superiority of our framework. Beyond these results, we observe that our method can be more efficient for safe on-line exploration, especially with multiple constraints. 

There are several limitations of our method which we will investigate in future research. For model-based learning, our algorithm only considers uncertainty in drift dynamics, which is not coupled with control input. However, uncertain input dynamics is also common in control systems. In terms of theoretic analysis, our criteria does not explain the quantitative impact of meta learning for probabilistic safety. Finally, although attractive in studied experiments, our algorithm needs more test on other applications to demonstrate its performance.

\section{Appendix}
\subsection{Detailed Illustration}
\label{appendix:illustration}
In the illustrative example, the moving point system is determined by \eqref{eqn:movingpoint} with $\omega=1.5$. The known safe trajectory, as the light gray line in Fig. \ref{fig:illustrationexample}, follows the a conic as $1.2 (x_1 - 1.5)^2 = x_2 - 0.3$. The samples of $x_1$ follow a uniform distribution $U(0.3, 1.5)$ independently. These samples are also used in experiments for warm start in Section \ref{sec:experiemnts}.

The safety boundaries plotted in Fig. \ref{fig:illustrationexample} are mathematically the zero-value contour of different worst-case robust CBF $h_a(x) = 0$ of a nominal system $\dot x = u + \varphi(x)$. Specifically, we follow the definition of the adaptive CBF in \cite{TaylorA20} as $h_a(x) = h(x) - \Delta(x)$, where $\Delta(x)$ is the worst-case estimation of $ \varphi(x)$. The red line represents the accurate estimation, while the green and orange lines are given by pessimistic estimation of different Bayesian models with $\beta = 1.96$. The shaded parts represent unreachable sets, different from the infeasible/unsafe set given by a classification model. In this vein, the boundaries imply different conservativeness of CBF methods.

\subsection{Proofs of Main Results}
\label{appendix:proof}
Proposition \ref{proposition:confidenceellipsoid} is derived by leveraging the ideas of vector learning in \cite{LewSHBP22} and linear stochastic bandits in \cite{YadkoriPS11}. The following additional definitions are needed for derivation.

We view the on-line sampling during the control process as generating sequence $\left\{\phi_t \right\}_{t=1}^\infty$ where $\phi_t = \phi(x_t)$ is the forward network in ABLR model. Consider the $\sigma$-algebra $\mathcal{F}_t = \sigma\left( \phi_1,\dots,\phi_{t+1}, \eta_1, \dots, \eta_{t} \right)$. Let us assume that $\left\{\mathcal{F}_t \right\}_{t=1}^\infty$ is any filtration such that for integer $t\ge 1$, $\phi_t$ is $\mathcal{F}_{t-1}$-measurable, $\eta_t$ is $\mathcal{F}_t$-measurable, and $\Delta_{\omega_{K+1}}(x_t) = \phi_t^\top \omega^* + \eta_t$ is also $\mathcal{F}_t$-measurable. Define a sequence $\left\{ S_t\right\}_{t=0}^\infty$ with $S_t = \sum_{s=1}^t \eta_t \phi_t$ as a martingale w.r.t $\left\{\mathcal{F}_t\right\}_{t=0}^\infty$.
\begin{lemma} [Self-Normalized Bound \cite{YadkoriPS11}]
    \label{lemma:selfnormalizedbound}
    Let Assumption \ref{assumption:noisebound} hold. For a $D\times D$ positive definite matrix $K_0$ and any integer $t\ge 0$, $K_t$ updates according to \eqref{eqn:updateonline}. Then, for any $\tilde\delta$ with probability at least $1-\tilde{\delta}$, there is
    \begin{equation}
        \Vert S_t \Vert^2_{K_t} \leq 2\sigma_0^2 \log\left( \frac{\det(K_t)^{-\frac{1}{2}} \det(K_0)^{\frac{1}{2}}}{\tilde\delta} \right), \forall t\ge0.
    \end{equation}
\end{lemma}

\begin{proof}[Proof of Proposition \ref{proposition:confidenceellipsoid}]
    With Assumption \ref{assumpstion:capacity}, The predicted mean of $\mathbf{w}$ can be rewritten as
    \begin{align}
        \bm{\mu}_t = & K_{t-1} \left( \Phi_{t-1} \mathbf{y}_{t-1} + K_0^{-1} \bm{\mu}_0\right) \nonumber\\
        = & K_{t-1} \left( \Phi_{t-1} (\Phi_{t-1}^\top \mathbf{w}^* + \bm{\eta}) + K_0^{-1} \bm{\mu}_0\right) \nonumber\\
        =& K_{t-1} \left( K_{t-1}^{-1} \mathbf{w}^* + \Phi_{t-1} \bm{\eta} + K_0^{-1} (\bm{\mu}_0- \mathbf{w}^*)\right)\nonumber\\
        =& \mathbf{w}^* + K_{t-1}\Phi_{t-1}\eta + K_{t-1}K_0^{-1}\left( \bm{\mu}_0 - \mathbf{w}^*\right),
        \label{eqn:proof_proposition1}
    \end{align}
    where $\bm{\eta}_{t-1} = \left[ \eta_1,\dots, \eta_{t-1} \right]^\top$ denotes the observation noise during sampling and estimation. Denoting $\mathbf{\tilde w}_i = \bm{\mu}_i - \mathbf{w}^*$, subtracting $\mathbf{w}^*$ on both side of \eqref{eqn:proof_proposition1}, and then left multiplying a vector $p \in \mathbb{R}^D$, we can get
    \begin{equation}
        \vert p^\top \mathbf{\tilde w}_t  \vert \leq   \Vert p \Vert_{K_{t-1}} \big(  \Vert \Phi_{t-1}\eta \Vert_{K_{t-1}} + \Vert K_0^{-1} \mathbf{\tilde w}_0 \Vert_{K_{t-1}} \big).
    \end{equation}
    The second term on the right-hand side is further scaled by
    \begin{equation}
        \Vert K_0^{-1}\mathbf{\tilde w}_0 \Vert^2_{K_{t-1}} \leq \frac{\lambda_{\max} ( K_{t-1} )}{\lambda_{\min}( K_0 )} \Vert \mathbf{\tilde w}_0 \Vert^2_{K^{-1}_{0}},
    \end{equation}
    By Lemma \ref{lemma:selfnormalizedbound} and Assumption \ref{assumpstion:calibration}, with probability at least $1-2\tilde\delta$ (see \cite{LewSHBP22}), the following inequality holds $\forall t\ge0$:
    \begin{equation}
        \Vert \mathbf{\tilde w}_t  \Vert_{K^{-1}_{t-1}}^2 \leq   \Vert K_{t-1}^{-1} \mathbf{\tilde w}_t  \Vert_{K_{t-1}} \Gamma^\delta_{t-1} =  \Vert \mathbf{\tilde w}_t  \Vert_{K^{-1}_{t-1}} \Gamma^\delta_{t-1} ,
    \end{equation}
    where we let $p=K_{t-1}^{-1} \mathbf{\tilde w}_t$. Furthermore, we have $\Vert \mathbf{\tilde w}_t  \Vert_{K^{-1}_{t-1}} \leq \Gamma^\delta_{t-1}$, implying $\mathbf{w}^*$ lies in the confidence ellipsoid $\mathcal{C}_t^\delta$.
\end{proof}


\begin{proof}[Proof of Theorem \ref{theorem:probabilitysafety}] 
With the on-line dataset $\mathcal{S}^t_{\omega_{K+1}}$, the predicted variance of a single test point $\hat x$ according to \eqref{eqn:predictedvarianceBLR} is
\begin{equation}
    \Sigma(\hat x) = \sigma^2_0 \left(1  + \phi(\hat x)^\top K_t \phi(\hat x)\right) \ge \sigma_0^2.
\end{equation}
Note that $K_t$ is positive definite according to \eqref{eqn:updateonline}. Therefore, we can derive the upper bound of $\Sigma_t$ by
\begin{equation}
    \Sigma(\hat x) \leq \sigma^2_0 + \sigma^2_0 \lambda_{\max}(K_t) \Vert \phi(\hat x) \Vert^2
\end{equation}
Thus, the predicted standard deviation satisfies
\begin{align}
    \tilde\sigma(\hat x) = \sqrt{\Sigma(\hat x)} \leq & \sigma_0 \sqrt{1 + \lambda_{\max}(K_t) \Vert \phi(\hat x) \Vert^2} \nonumber \\
    \leq & \sigma_0 \left(1 + \sqrt{\lambda_{\max}(K_t)}\Vert \phi(\hat x) \Vert  \right). \label{eqn:theoremsigmabound}
\end{align}
Replacing $\tilde\sigma_t$ in \eqref{eqn_pcbfQP_cbf} by the above upper bound, we have
\begin{align}
    & \mathscr{L}_f h_\omega(\hat x)  + \mathscr{L}_g h_\omega(\hat x) u  + \tilde\mu(\hat x) \nonumber\\ 
    & - \beta \sigma_0 \left(1 + \sqrt{\lambda_{\max}(K_t)}\Vert \phi(\hat x) \Vert  \right) + \alpha \left(h_\omega(\hat x)\right) \ge 0. \label{eqn:theoremuppersigmaCBF}
\end{align}
Note that satisfying \eqref{eqn:theoremuppersigmaCBF} is sufficient to also meet \eqref{eqn_pcbfQP_cbf} for any input $u$. We next study the relationship between the upper bound of standard deviation \eqref{eqn:theoremsigmabound} and the probabilistic error bound in \eqref{eqn:weightupperbound}.

On the other hand, according to the adaptive CBF methods \cite{TaylorA20, BlackAP21, WangLWSZH23}, the upper error bound in \eqref{eqn:weightupperbound} can be considered in vanilla CBF method to ensure robustness against uncertainty. This results in the following adaptive CBF constraint:
\begin{align*}
     \mathscr{L}_f h_\omega(x_t)   &+ \mathscr{L}_g h_\omega(x_t) u   + \bm{\mu}_t^\top\phi(x_t) \nonumber \\
    & -  \sup_{\mathbf{w} \in \mathbb{R}^D}\Vert \mathbf{w} - \bm{\mu}_t \Vert \Vert \phi(x_t) \Vert + \alpha \left(h_\omega(x_t)\right) \ge 0.
\end{align*}
Note that we do not consider the variation of this upper bound for better adaptation as \cite{TaylorA20, BlackAP21, WangLWSZH23}, but only a conservative upper bound as the worst-case CBF constraint \cite{Jankovic18}. Note also that the worst-case CBF constraint is sufficient to ensure safety. According to Proposition \ref{proposition:confidenceellipsoid}, with probability at least $1 - \tilde \delta$, 
 \begin{equation}
      \frac{1}{\sqrt{\lambda_{\max}(K_t)}} \Vert \mathbf{w}^* - \bm{\mu}_t \Vert \leq \Vert \mathbf{w}^* - \bm{\mu}_t \Vert_{K_t^{-1}} \leq  \sigma_0 \Gamma_t^{\tilde \delta}.
\end{equation}
Therefore, we have with probability at least $1 - \tilde \delta$, the following constraint is sufficient to render safety:
\begin{align*}
    \mathscr{L}_f h_\omega(x_t) & + \mathscr{L}_g h_\omega(x_t) u   + \bm{\mu}_t^\top\phi(x_t) \nonumber \\
    & - \sigma_0 \sqrt{\lambda_{\max} (K_t)}\Gamma_t^{\tilde \delta} \Vert \phi(x_t) \Vert + \alpha \left(h_\omega(x_t)\right) \ge 0.
\end{align*}
Comparing this with \eqref{eqn:theoremuppersigmaCBF}, we can take $\beta \ge \Gamma_t^{\tilde \delta}$ to obtain that \eqref{eqn:theoremuppersigmaCBF} renders safety  with probability at least $1 - \tilde \delta$.
\end{proof}

\section*{Acknowledgment}
This work was supported by NPRP grant: NPRP 9-466-1-103 from Qatar National Research Fund, UKRI Future Leaders Fellowship (MR/S017062/1, MR/X011135/1), NSFC (62076056), EPSRC (2404317), Royal Society (IES/R2/212077) and Amazon Research Award.
{
\bibliographystyle{IEEEtran}
\bibliography{IEEEabrv, mybib}
}

\end{document}